\documentclass[twocolumn]{article}

\usepackage{lmodern}
\usepackage{xcolor}
\usepackage[margin=1.0in]{geometry}
\usepackage{amssymb,amsmath}
\usepackage{parskip}
\usepackage{pifont}
\usepackage{algpseudocode}
\usepackage{float}
\usepackage{ifxetex,ifluatex}
\usepackage{optidef}
\usepackage[htt]{hyphenat}
\usepackage[colorlinks=true,
            linkcolor = black,
            urlcolor  = blue,
            citecolor = blue,
            anchorcolor = black]{hyperref}
\usepackage[utf8]{inputenc}
\usepackage{xcolor}
\usepackage{MnSymbol}
\usepackage{amsmath}
\usepackage{mathtools}
\usepackage{amsfonts}
\usepackage{enumerate}
\usepackage{amsthm}
\usepackage{comment}
\usepackage{dsfont}
\usepackage{pgfplots}
\usepackage{biblatex}
\usepackage{pgf-pie} 
\usepackage{algorithm}
\usepackage{algpseudocode}
\usepackage{tikz}
\usepackage{caption}
\usepackage{subcaption}
\usepackage{tikz-cd}
\usepackage{tabulary}
\usepackage{booktabs}

\usepackage[skins,breakable]{tcolorbox}
\graphicspath{ {./images/} }
\usepackage{float}
\newtcolorbox{mybox}[1][]{enhanced jigsaw,breakable,pad at break=1mm,
  oversize,left=8mm,interior hidden,colframe=black,nobeforeafter=,#1}

\usepackage[edges]{forest}

\theoremstyle{definition}
\newtheorem{theorem}{Theorem}[section]
\newtheorem{lemma}[theorem]{Lemma}
\newtheorem{definition}[theorem]{Definition}

\newcommand{\cmark}{\ding{51}}%
\newcommand{\xmark}{\ding{55}}%

\author{
  \textbf{Bruno Mazorra}\\
  Nokia Bell-labs\\
  Universitat Pompeu Fabra\\
  \texttt{brunomazorra@gmail.com}\\
  \and
\textbf{Michael Reynolds}\\
  University College  London \\
  \texttt{mireynolds@pm.me}
  \and
\textbf{Vanesa Daza}\\
  Universitat Pompeu Fabra \\
  \texttt{vanesa.daza@upf.edu}
} 
\title{Price of MEV: \\ Towards a Game Theoretical   Approach to MEV}

\begin{document}

\maketitle

\begin{abstract}
Maximal (also miner) extractable value, or MEV, usually refers to the value that \textit{privileged} players can extract by strategically ordering, censoring, and placing transactions in a blockchain. Each blockchain network, which we refer to as a domain, has its own consensus, ordering, and block-creation mechanisms, which gives rise to different optimal strategies to extract MEV. The strategic behaviour of rational players, known as searchers, lead to MEV games that have different impacts and externalities in each domain. Several ordering mechanisms, which determine the inclusion and position of transactions in a block, have been considered to construct alternative games to organise MEV extraction, and minimize negative externalities; examples include sealed bid auctions, first input first output, and private priority gas auctions. However, to date, no sufficiently formal and abstract definition of MEV games have been made.
In this paper, we take a step toward the formalization of MEV games and compare different ordering mechanisms and their externalities. In particular, we attempt to formalize games that arise from common knowledge MEV opportunities, such as arbitrage and sandwich attacks. In defining these games, we utilise a theoretical framework that provides groundwork for several important roles and concepts, such as the searcher, sequencer, domain, and bundle. We also introduce the price of MEV as the price of anarchy of MEV games, a measure that provides formal comparison between different ordering mechanisms. 

\end{abstract}


\maketitle
\section{Introduction}

The notion of miner-extractable value (MEV) introduced in \cite{daian2020flash}, and formally defined in \cite{babel2021clockwork}, measures the value that miners can extract by strategically ordering, censoring, and including transactions in a block. In general, block proposers (typically validators or miners), have the power to dictate or influence the inclusion and ordering of pending transactions in a new block. Thus, rational proposers have access to extra rewards per block. However as shown in \cite{daian2020flash}, this power is not limited to block proposers. Blockchain users and bots, usually termed \textit{searchers}, can use strategies like spamming transactions or outbidding competitors to extract MEV themselves. For this and other reasons, the term miner-extractable value has evolved to mean \textit{maximal extractable value}, formally defined in \cite{MaxEV}, which does not limit extraction to block proposers. We often refer to block proposers and searchers as players.

Depending on the network protocol, which we refer to as the domain, and type of MEV being extracted, rational searchers have incentives to take different actions to increase their expected revenue. For example, \cite{daian2020flash} showed that Ethereum, a network with higher propagation latency, incentivises searchers to widen their view of the mempool to outbid their competitors, while domains like Polygon PoS, where the propagation of transactions is probabilistic, incentivises searchers to spam transactions to extract MEV. This motivates a more general formulation of the MEV game presented in \cite{daian2020flash} such that models of rational searcher behaviour can be derived from a domain input.

\vspace{-0.3cm}
\subsection{Our contributions} 

The main goal of this paper is to formalise MEV games with a finite number of players under different ordering mechanism and auction designs which are defined by various environment conditions and rules such as latency of players, PoW/PoS consensus, and mempool visibility. To this end, the following topics will be covered: 
\begin{itemize}
    \item Inspired by \cite{babel2021clockwork}, we provide a formal definition of MEV and local MEV as an optimization problem constrained by player resources. In our definition, we emphasize relative constraints like capital, software resources, and view of the mempool.
    \item We formally define a more abstract and general MEV game that provides a foundation for analysing the incentives of searchers in different domains. 
    \item We also introduce different negative externalities that arise from MEV games, and we formally define the Price of Anarchy in MEV games under different cost functions as a way to quantify their impact. We explore the specific case of block space misuse, which we refer to as the price of MEV, and compute this to characterize and classify different ordering mechanisms.
    \end{itemize}
    
\subsection{Previous Work} 

In Flashboys 2.0 \cite{daian2020flash} Daian et al. proposed a formalization of the priority gas auction (PGA) model and proved mathematically and empirically the existence of Grim-Trigger Nash equilibrium under certain conditions. However, the results are limited to two players and the PoW Ethereum domain, and do not consider reverted transaction costs proportional to the bid or  that players can engage in Sybil attacks. In \cite{kulkarni2022towards}, the authors explore the equilibrium of sandwich MEV as a router game. Likewise, in \cite{weintraub2022flash}, the authors compare empirically the miner's revenue and the searchers competition of arbitrage opportunities before and after the introduction of Flashbots MEV auctions on Ethereum. In \cite{obadia2021unity}, the first formalization of cross chain MEV was proposed. However, to the best of our knowledge, there is no formal and empirical study on the impacts of MEV on several distinct domains.


\section{Theoretical Framework}\label{theoryFramework}

In general, miner or maximum extractive value (MEV) is a term that refers to any excess profits that a block proposer or searcher can make based on transaction ordering and/or transaction inclusion.
MEV was introduced in \cite{daian2020flash}, formally defined in \cite{babel2021clockwork}, and extended in cross-domain environments in \cite{obadia2021unity}. Another similar definition was given in \cite{MaxEV}, except MEV was treated independently of a player. The MEV opportunities we consider in this paper are limited to the profits a player can obtain by modifying the blockchain state. In this section, we will formalize this kind of MEV opportunity using the concept of profitable `bundles'. Then, we will define `local' MEV as the maximally profitable bundle a specific player can construct. Similar to \cite{obadia2021unity}, we start by formally defining the domain and searcher:
\begin{definition}
A \textit{domain} $\mathcal D$ is a self-contained system with a globally shared state $\texttt{st}$. This state is altered by
various agents through actions (sending transactions, constructing blocks, slashing, etc.), that execute within a shared execution
environment’s semantics. Each domain has a predefined consensus protocol that includes a set of valid algorithms to order transactions, denoted by $\texttt{prt}(\mathcal D)$. 
\end{definition}
A blockchain is a domain, however, there are other non-blockchain domains that also have MEV, like centralized exchanges.
\begin{definition} A \textit{searcher} (in general, we will call it a player) in a domain $\mathcal D$  is a participant that assumes that sequencers follow a specific set of rules and take strategic actions (send bundles with specific bids) to maximize their own utility. In general, we will assume that a player's utility depends linearly on their token balances.
\end{definition}

In a domain $\mathcal D$ with state $\texttt{st}$, the update of the state $\texttt{st}$ after executing transactions $\texttt{tx}$ is given by $\texttt{st}\circ \texttt{tx}$. For an ordered set of transactions $B=\{\texttt{tx}_1,...,\texttt{tx}_l\}$, we have the composition $\texttt{st}\circ B = \texttt{st}\circ \texttt{tx}_1\circ....\circ \texttt{tx}_l$.

Similar to \cite{babel2021clockwork}, we use $\texttt{Addr}$ to denote the set of all possible accounts and $\textbf{T}$ to denote the set of all tokens. We define $b:A\times \textbf{T}\rightarrow \mathbb Z$ as the function that maps a pair of, an account and a token, to its current balance. More precisely, for $a\in \texttt{Addr}$ we let $b(a,\cdot)$ denote the balance of all tokens held in $a$ and $b(a,T)$ denote the account balance of token $T$. 
Abusing notation, we will denote by $b(a)$ the value of $b(a,\cdot)$ priced by a numéraire $E$. That is, if there is a pricing vector $p=(p_{T\rightarrow E})_{T\in\textbf{T}}$, then $b(a)= p\cdot b(a,\cdot)=\sum_{T\in\textbf{T}}p_{T\rightarrow E}b(a,T)$.
\begin{definition} An \textit{ordering mechanism} is a set of rules that determines the order and inclusion of a set of transactions in a block. More formally, let $\mathcal T$ be the set of all transactions, an ordering mechanism is a map $\textbf{or}:\mathcal P^{\leq}(\mathcal T)\rightarrow \mathcal P^{\leq}(\mathcal T)$, where $\mathcal P^{\leq}(\mathcal T)$ is the set of all ordered subsets of $\mathcal T$, such that $\textbf{or}(T)\subseteq T$ for all $T\in \mathcal P^{\leq}(\mathcal T)$.
\end{definition}
\begin{definition} A \textit{sequencer} is an agent of a domain responsible for maintaining the liveness and consistency through a set of actions. We distinguish four types of sequencers: \emph{dummy}, \emph{dummy Byzantine}, \emph{rational} and \emph{partially rational}. A sequencer is \textit{dummy} if he  follows the validator consensus protocol $\texttt{prt}(\mathcal D)$. A sequencer is \textit{ dummy Byzantine} if they misbehave, but other nodes can detect his misconduct. A sequencer is \textit{rational} if it follows a set of valid actions on the domain $\mathcal D$ to maximize its revenue (including deviating from an ordering mechanism). Therefore, if a player misbehaves to maximize their payoff but can not be identified, punished, or slashed, we say that is a rational player. A sequencer is \textit{partially rational} if they commit to using a specific valid ordering mechanism to maximize its payoff. 
\end{definition}

On Ethereum, miners are usually partially rational. In general, miners run \texttt{mev-geth}, which receives a block transaction ordering of highest revenue from the Flashbots relayer. Rational sequencers could take the Flashbots block and reorg to maximize their own revenue. Nevertheless, miners do not deviate from the \texttt{mev-geth} intended protocol\footnote{However, recent miners are deviating from the Flashbots protocol to extract more value. No formal or academic work has proved this yet. These are statements from some Flashbots team members.}. For this reason, players participating in the MEV extraction are not necessarily sequencers; they take a sequence of actions to bias nodes to maximize their own revenue.
The set of actions that a sequencer can follow depends on the domain and can include arbitrary actions. Studying the impact that rational players and sequencers can have in a domain is helpful to bound the set of actions. For example, in selfish mining \cite{sapirshtein2016optimal} the set of actions to construct or publish the private chains. 

From now on, we will assume that sequencers are partially rational. An example of an ordering mechanism is the default \texttt{geth} client, which uses a greedy approximation algorithm to optimize the blocks' transaction fee revenue.

A sequencer receives a set of concurrent transactions $\texttt{tx}_1,...,\texttt{tx}_n$ with gas price $m_1,...,m_n$ and $g_1,...,g_n$ units of gas. If the sequencer includes $\texttt{tx}_i$, it obtains $m_ig_i$ in fees. Since the gas used per block is restricted in every domain by some constant $L$, the sequencer must choose a subset of transactions $\mathcal T$ such that $\sum_{i: \texttt{tx}_i\in\mathcal T} g_{i}\leq L.$ Then, a node that tries to maximize its revenues per block needs to solve the following Knapsack optimization problem, which we name Knapsack Extractable Value (KEV) problem:

\vspace{-0.5cm}
\begin{maxi*}|s|
         {}{\sum_{i=1}^n x_im_ig_i}
         {}{}
         \addConstraint{\sum_{i=1}^nx_ig_i\leq L}{}
         \addConstraint{x_i\in\{0,1\}}{}.
\end{maxi*}
We note by KEV$(\mathcal T)$ the solutions' revenue of the optimization problem. Knapsack optimization problems \cite{knapsack} are NP-complete; that is, no known polynomial time algorithm finds an optimal solution. Thus, each sequencer usually chooses different algorithms to approximate the optimal solution. For example, Parity nodes order transactions by gas price $m$ without considering the gas costs. 

Another example is the \hyperlink{https://docs.flashbots.net/flashbots-auction/releases/alpha-v0.5}{Flashbots relayer}, which uses a greedy approximation algorithm \cite{akccay2007greedy}, ordering transactions by the ratio of miner payment and gas consumed (a natural extension of ordering by gas price taking into account direct payments and more than one transaction). However, the problem that the Flashbots relayer tries to solve is quite different, since it does not include reverted transactions or competing bundles. Moreover, we proved (see \ref{appendix:counter}) that there are examples where this algorithm does not produce a good approximation. In \cite{angeris2021note}, the authors reformulate the block production as a linear programming problem without taking into account competing bundles. We will give more details after defining bundles.

Now that we have settled up some ground definitions, we will focus next on the formalization of the extraction of MEV opportunities. To simplify the games, we will not consider the complete MEV extraction per block, but separate the MEV opportunities into ``independent/concurrent'' ones.

\begin{definition} Let $\mathcal P$ be the set of all  players. A set of transactions $\mathcal T=\{\texttt{tx}_1,...,\texttt{tx}_k\}$ and a state $\texttt{st}$ induce an MEV opportunity in a domain $\mathcal D$ to a player $P\in\mathcal P$ if they can construct an ordered set of transactions $B$ such that:
\begin{equation}
    \Delta b(P;\texttt{st}\circ B,\texttt{st}) := b(P,\texttt{st}\circ B)-b(P,\texttt{st}) > 0,
\end{equation}
where $b$ is the balance of $P$ with the corresponding order. We call $B$ a \textit{profitable bundle} or \textit{bundle}. If $B$ consists of a unique transaction, we say that $B$ is an \textit{MEV-transaction}. Each bundle can contain extra metadata, such as sequencer timestamps, bundle hash, sender ID, and gas price. For a given state $\texttt{st}$, each bundle incurs execution costs called gas costs. From the bundle metadata and the domain state, the bundle execution incurs some payments, denoted by $\textbf{pr}(B)$, to a sequencer or set of sequencers responsible for executing the bundle.
\end{definition}
\begin{definition} A set of bundles $B_1,...,B_n$ are \textit{order-invariant valid} if for every permutation $\sigma\in S_n$ we have that the state transition
\begin{equation}
    \texttt{st} \longrightarrow \texttt{st} \circ B_{\sigma(1)}\circ...\circ B_{\sigma(n)}
\end{equation}
is a valid state transition and is invariant among all the permutations. A set of bundles $B_1,...,B_n$ \textit{compete} in a state $s$ if for all $i$ and $j$, $\texttt{st}\rightarrow \texttt{st}\circ B_i\circ B_j$ is not a valid state transition.
\end{definition}
\begin{definition}  We say that a bundle $B$ is a partial extraction of a bundle $B'$ if $B$ and $B'$ compete, and $\Delta b(P,\texttt{st}\circ B)<\Delta b(P',\texttt{st}\circ B')$.
\end{definition}

The Flashbots combinatorial auction (FBCA) allows players to bid for bundles. The Flashbots allocation rule tries to solve the block optimization problem with conflicting constraints. That is, the block can not contain competing bundles (bundles that contain same transactions or bundles that revert). 
So, the FBCA can be modelled as the knapsack problem with a conflict graph $G=(V,E)$ \cite{pferschy2009knapsack}, where $V$ is the set of bundles, and $uv\in E$ if and only if bundles $u$ and $v$ compete. Flashbots use a greedy approximation algorithm, leading to an auction mechanism similar to a first-price sealed auction, since players can only observe winning bundles.  That is, the bundles are ordered by effective gas price (or average gas price, see more details in the appendix \ref{appendix:order}) and afterwards prune the conflicting bundles. In case of symmetric gas efficiencies, the MEV opportunity is sealed to the higher bidder and pays what they bid.
In general, this algorithm does not give the optimal solution (see appendix \ref{appendix:counter}). It also allows searchers to check for relayer deviation with just the executed block and the bundle sent by the searcher. In other words, theoretically, searchers can privately monitor the correct functioning of the relayer.

Now we are ready to define the local MEV for a player $P$ or  $\text{MEV}_{P}$ for short. The definition we provide is similar to the one provided in \cite{babel2021clockwork}. However, in \cite{babel2021clockwork}, players have constraints on the state transitions, but not on the set of bundles that they can construct.
\begin{definition} Let $\mathcal D$ be a domain with state $\texttt{st}$, a player $P$ with local mempool view $\mathcal T^M_P$ and a set of transactions $\mathcal T_P$ that the player $P$ can construct. We denote by $\mathcal C_P=\mathcal T^M_P\cup \mathcal T_P$ to be the set of reachable transactions. We define the \textit{local} MEV \textit{of $P$ with state $\texttt{st}$ } ($\text{MEV}_P(\texttt{st})$) as the solution to the following optimization problem
\begin{maxi*}|s|
{B}{\Delta b(P;\texttt{st}\circ  B,\texttt{st})}
{}{}
\addConstraint{ B\subseteq \mathcal C_P}{}
\addConstraint{\texttt{st}\rightarrow\texttt{st}\circ  B\text{ is a valid state transition in $\mathcal D$}}{}
\end{maxi*}
Let $\text{argmev}_P(\texttt{st})$ be the set of bundles that are a solution to the optimization problem\footnote{Observe that this definition fundamentally depends on the token balance. In  the presence of a unique token, $\text{MEV}_P$ is trivially defined. Nevertheless, in the presence of multiple domains and tokens this definition is non-trivial. Moreover, in this definition, we are assuming that all players value equally the tokens, assuming the existence of some transferable utility. We leave a more general definition of local MEV for future work.}. The constraints of reachable bundles is subject to a player's information, gas efficiency, budget, ability to propose blocks, etc.

Observe that if all players have access to the same MEV opportunities and have access to all bundles, then  this definition is equivalent to the one provided in \cite{babel2021clockwork}. If the mempool view is the empty set, we will refer to the MEV as \textit{on top of block} MEV, and we will denote it by $\text{TMEV}_P$.
\end{definition}
\begin{lemma} For a given player $P$ and a state $\texttt{st}$, if \\$B\in \text{argmev}_P(\texttt{st})$, then $\text{MEV}_P(\texttt{st}\circ B)=0$.
\end{lemma}
\begin{proof} Assume otherwise. Let $\texttt{st}$ and $B\in\text{argmev}_P(\texttt{st})$ such that $\text{MEV}_P(\texttt{st})>0$. Then, there exist $B'$ such that $\Delta b(P,\texttt{st}\circ B\circ B',\texttt{st}\circ B)>0$. Taking the bundle $B''=B\cup B'$, we have that $\Delta(P,\texttt{st}\circ B'',\texttt{st})>\text{MEV}_P(\texttt{st})$, that is a contradiction.
\end{proof}
\begin{definition} For a set of reachable bundles $\mathcal C$, we define the \textit{ $\mathcal C$-permissionless} MEV ( $\text{MEV}^\mathcal C$) as the minimum local MEV that can be extracted among players that have access to the bundles in $\mathcal C$. More formally,
\begin{equation*}
     \text{MEV}^{\mathcal C}(\texttt{st}):=\min_{P\text{ s.t }\mathcal C\subseteq \mathcal T_P}\text{MEV}_P(\texttt{st}).
\end{equation*}
If for all players $P$ with reachable bundles $\mathcal C$, $\text{MEV}_P(\texttt{st})=0$, we say that $\texttt{st}$ is a null MEV state. We denote the set of null MEV states as NS.
\end{definition}
\textbf{Remark}: Note that this definition is an extension of the definition made in \cite{MaxEV} when $\mathcal C_K$ is the set of transactions that burns $K$ coins. In other words, $\text{MEV}^{\mathcal C_K}$ permissionless is MEV that can be extracted by players with at least $K$ coins. In this paper, we will assume that a finite number of players are able to capture a specific MEV opportunity. In other words, we will fix a set of players $\mathcal P = \{P_1,...,P_n\}$. For each $P\in\mathcal P$, we define $\texttt{st}_i$ as the state that realizes the extraction $\text{MEV}_P$. Then, we will assume that for each player $P\in\mathcal P$,  $\text{MEV}_P(\texttt{st}_j)=0$ for all $P\in\mathcal P$. 
\section{The MEV stage game} \label{section:abstract_game}
In this section, we will consider an abstraction of the PGA model proposed in  \cite{daian2020flash}. We will model and formalize the MEV stage game. In other words, we will formalize the game played for extracting a given MEV opportunity in a specific block. Afterwards, we will define the utility of the players as a function of their balance and the notion of strategy. Finally, we introduce a solution concept of the MEV stage game, the Sybil resistant Nash equilibrium.

We model the MEV game as a sequential game among a set of $n$ searchers $\mathcal P=\{P_1,...,P_n\}$ who can send bundles to obtain an MEV opportunity. We will assume that all players compete for the same MEV opportunity and have sufficient capital to extract it. When a specific player wins the MEV opportunity, it reaches a null MEV state for all players. In the following, we will provide a list of points that will define the MEV stage game. This game will take into account the latency of the players, the duration of the blocks, the mechanism of transaction inclusion, and the costs of improving software, node location, etc.
\begin{enumerate}
    \item  \textbf{Continuous time}: Searchers act in continuous time rather than discrete rounds (as in typical extensive-form games). That is, at any moment in time, players can take an action that, in our case study, will be sending a bundle.
    \item \textbf{Local MEV}: At time $t$, each player $P_i$ finds an MEV opportunity of value $v_i(t)\sim V_i(t)$, (in general, such that $v_i(t)\geq v_i(s)$) for all $t\geq s$, with $\{V_i(t):t\geq0\}$ being a  family of distributions. More specifically, for each time $t$, the player finds a bundle $B$ such that $\Delta b(\texttt{st}\circ B,\texttt{st}) = v_i(t)$.
    \item \textbf{Latency}: Searchers can see each other's actions, but not immediately, due to the latency in the peer-to-peer network. The latency is modeled by a directed weighted graph $G=(E,V)$. Each searcher controls a set of nodes $N_i\subseteq V$. So, if $ P_j$ sends a bundle from a subset of peers $L_j\subseteq N_j$ at time $t_j$, $ P_i$ observe the bundle at time $t=t_j+d(L_j,N_i)$, where $d$ is the non-symmetric distance induced by the weight.
    \item  \textbf{Probabilistic auction duration}: The auction terminates at a randomly drawn time when a new block is mined. We model the block interval as a positive random variable $\mathcal B$. 
    \item \textbf{Competitors information}: Players do not necessarily know the number of competitors and their features. However, each player estimates the number of competitors and their behavior.
    \item \textbf{Access to a public correlating device}: Let $(\Omega,\Pr)$ be a probability space. All searchers observe the first drawn $w\in [0,1]$ of a uniform public random variable $X$ (a beacon or the hash of the previous block). This can be used by players to coordinate their actions.
    \item \textbf{Auction Mechanism}: Sequencers have a predefined algorithm that inputs a set of bundles and outputs an order of transactions for inclusion in a block. We will denote this algorithm by the \textit{ordering mechanism}. This ordering mechanism and the characteristics of the MEV opportunity (Back-running, Front-running, sandwich,...) determines the revenue for each player. 
    The ordering mechanism and the MEV opportunity determinate an auction $\mathbb A=(\textbf{x},\textbf{pr})$, that is, a pair of maps that take as inputs the set of bundles and a random event and outputs a winner and the payment induced per each player. More formally,
    \begin{align*}
        \textbf{x}:\texttt{View}_s\times\Omega&\rightarrow \mathcal\{x\in\{0,1\}^n:x\cdot\textbf{1}^T\leq 1\}\\
    \textbf{pr}:\texttt{View}_s\times\Omega&\rightarrow \mathbb R^n.
    \end{align*}
    where $\texttt{View}_s$ is the set of transactions seen by the sequencer when constructing the block. 
    \item \textbf{External costs}: Players can improve their mempool view, reduce their latency, and improve their software to increase their local MEV. The set of external actions that a player $P_i$ can take rely on a set of actions $\text{A}_i$, and the costs of taking those actions are modeled by a function $c_i:\text{A}_i\rightarrow \mathbb R$.
\end{enumerate}
\begin{definition} Given a set of players $\mathcal P$, a random positive variable $\mathcal B$, a gossip network graph $G$ and an auction mechanism, $\mathbb A=(\textbf{x},\textbf{pr})$, an MEV game is $\mathbb G = (\mathcal P,G, \mathcal B,\mathbb A,(c_i)_{i=1,..,n})$. We say the game is symmetric if all players share the same features.  That is, $V_i=V_j$, $c_i=c_j$ for all $i,j$, $\mathbb A$ are symmetric functions and $G$ is a homogeneous graph.
\end{definition}
The MEV game has a structure of sequential continuous game, and therefore we can define utilities, strategies, and solution concepts such as Nash equilibrium.
\begin{definition} A strategy $S_i$ is a procedure for participating in the MEV stage game and may be probabilistic. $S_i$ takes the following form, for a current time $t$ and a local view $\texttt{view}_i(t)$ of the player $i$:
$    (a,\texttt{view}_i'(t))\leftarrow S_i(t,\texttt{view}_i(t)).$
The output $a$ is the action taken by $P_i$ and is bounded by domain constraints.
The output $\texttt{view}_i'(t)$ is the updated state, that is, $\texttt{view}_i'(t)=\texttt{view}_i(t)\cup\{a\}$. A strategy of a player $i$ is non-adaptive if it does not depend on the local view $\texttt{view}_i(t)$. More formally, for every $t$ and every pair $\texttt{view}_i(t),\texttt{view}_i'(t)\in \texttt{View}_i(t)$, it holds $S_i(t,\texttt{view}_i(t))=S_i(t,\texttt{view}_i'(t))$. 
\end{definition}

\begin{definition}
Let $S = (S_1,...,S_n)$ be a strategy tuple, then the \textit{expected payoff of the player $P_i$} is 
\begin{equation*}
    u_i(S_i,S_{-i}) := \mathbb E [\Delta b_i\mid S],
\end{equation*} 
where $S_{-i}$ is an $n-1$ tuple without the $i$th coordinate. We denote by $\mathbb S_i$ the set of all strategies.
\end{definition}
 
That is, $u_i(S_i,S_{-i})$ is the expected payoff of player $i$ if they execute the strategy $S_i$, the other players execute the strategy $S_{-i}$ in a domain $\mathcal D$.   Notice that we assume players have a monotone risk-neutral utility over the balances. 

\begin{definition} Let $\mathbb G$ be an MEV stage game, we say that a tuple of strategies $(S_1,...,S_n)$ is a \textit{Nash equilibrium} (NE) if for each player $P_i$
\begin{equation*}
    u_i(S_i,S_{-i}) \geq u_i(\tilde{S_i},S_{-i}),\text{ for all strategies $\tilde{S_i}$}.
\end{equation*}
In other words, if players are taking the strategies $(S_1,...,S_n)$, none of them have incentives to deviate unilaterally. We denote NE($\mathbb G$) the set of all Nash equilibrium.
\end{definition}

However, the notion of Nash equilibrium is not strong in permissionless environments. For example, the equilibrium of firms competing in the Cournot price model \cite{ruffin1971cournot} is weak against Sybil attacks. Another example is the equilibrium constructed in the theorem of Flashboys 2.0 \cite{daian2020flash} that proves that exponential raise bidding strategies with grim-trigger area a Nash equilibrium in the game with two players. Nevertheless, one can prove that agents have incentives to use Sybil attacks to maximize their payoff. For this reason, we introduce a stronger notion of a solution concept for MEV games that are resistant to Sybil attacks:
\begin{definition}
Let $\mathbb G$ be a symmetric\footnote{This definition can be extended to non-symmetric games, but we will leave it for future work.} MEV game with $n$ players, and let $\phi :\mathbb N\rightarrow \bigoplus_{i=1}^\infty \mathbb S_i$ be a strategy mapping. A Sybil resistant Nash equilibrium is $\phi$ such that for all $n\in\mathbb N$, $\phi(n)\in \bigoplus_{i=1}^n \mathbb S_i$, and $\phi(n)$ is a Nash equilibrium, where for each $i\in \{1,\ldots,n\}$ and $j\in \{1,\ldots,n+1\}-\{i\}$,
    \begin{align*}
        u_i(\phi(n)_i,\phi(n)_{-i})\geq\; &u_i(\phi(n+1)_i,\phi(n+1)_{-i})\\
        &+u_j(\phi(n+1)_j,\phi(n+1)_{-j}).
    \end{align*}
\end{definition}
The notion of Sybil resistance is important in permissionless pseudo-anonymous environments such as blockchains. Players have the ability to generate additional addresses to take more profits from cooperative strategies. In future work, we will prove that Sybil resistant cooperative Nash equilibrium exist in the priority gas auction in the non-repeated games, and the existence of Sybil resistant equilibrium in games with private mempools.

\section{Price of MEV}\label{section:pricemev}
MEV games have an important impact on users, network congestion, computation overload, and blockchain liveness. On one hand, some MEV opportunities arise from value extracted from users. On the other hand, in general, MEV games induce an inefficient extraction of MEV opportunities, leading to network congestion (e.g. P2P network load), and chain congestion (e.g. block-space usage). In this paper, we will not take into account the negative externalities of MEV on individual users (for example, sandwich attacks and oracle manipulation), but rather the negative externalities that impact the consensus protocol, liveness, chain quality, stake distribution,  etc. Suppose there is an MEV opportunity on a state $\texttt{st}$, with a set of $n$ players that compete to extract it, sending bundles $B_1,...,B_k$. Then, a sequencer, using the ordering mechanism $\textbf{or}$, outputs a block (this order mechanism does not necessarily respect the internal order of the bundles). In this setting, we define the \textit{block space cost} as the gas cost of executing the block built by the ordering mechanism using the bundles, more formally 
\begin{equation}
   C(B_1,...,B_k;\texttt{st})=\text{gasUsed}\left(\texttt{st}\circ \textbf{or}(\cup_{i=1}^k B_i)\right). 
\end{equation}
We naturally extend the function of cost over strategies by taking the outcome (or the expected outcomes in case of mixed strategies) of the strategies (the broadcasted bundles).
Clearly, ex-post it is trivial to compute the gas cost. However, ex-ante, estimating the gas costs induced by the MEV game is much more complex due to the non-commutative nature of execution costs. Also, other cost functions can be very relevant in some domains. Different MEV games can induce other types of negative externalities, such as wasted resources induced by computation costs\footnote{For example, the energetic costs induced by computing TH/s in blockchains that order transactions by nonce. See \hyperlink{https://github.com/bnb-chain/bsc/pull/915}{BSC-PR} for more details.} or centralization effects.

In general, self-interested behaviour by strategic players leads to an inefficient result, an outcome that could be improved upon given centralized control over everyone's action \cite{roughgarden2015intrinsic}. Nevertheless, imposing such control can be costly, infeasible or undesirable (due to trust assumptions). This motivates the search for conditions and mechanisms in which decentralized optimization by strategic agents is guaranteed to produce a near-optimal outcome. The price of anarchy (PoA) \cite{roughgarden2005selfish} is a measure that quantifies how far is the worst Nash equilibrium (in the sense of social cost) with respect to
any optimal configuration that minimizes the social cost. More formally, given a cost function $C$ and the set of Nash equilibrium $NE$, the price of anarchy is defined as:
\begin{equation*}
    \text{PoA} = \frac{\max_{S\in NE}C(S)}{\min_{S}C(S)}.
\end{equation*}
Different examples of the study of the price of anarchy can be seen in \cite{christodoulou2005price,roughgarden2015intrinsic,roughgarden2005selfish}. However, the price of anarchy in the MEV game is, in general, not well-defined. For example, assume that two players are competing for extracting the same arbitrage opportunity. Then, as we will see, the block space cost of extracting the arbitrage opportunity will be the sum of gas used by executing both searchers' transactions. However, the minimal cost of the game is zero, since not extracting the MEV opportunity is a feasible outcome. Therefore, the ratio is not defined, leaving an inconsistent definition of the price of anarchy in the MEV game. In the following, we propose a small adjustment to have a well-defined price of anarchy in the MEV game, which we denote as the Price of MEV. This measure tracks the social costs induced by the competition among individually rational agents for MEV extraction in a particular MEV game. More precisely, the Price of MEV is a family of measures parametrized by the social cost functions. This family of measurements can be useful to compare the negative externalities and trade-offs of different MEV games. 
Similar to the price of anarchy definition, the price of MEV of game $\mathbb G$ with social cost $C$ is the ratio of the worst Nash equilibrium with respect to the extraction made by the most efficient player in an order-free consensus protocol blockchain.   

\begin{definition} Given a cost function $C$, the set of Sybil resistant Nash equilibrium SNE($\mathbb G$), and the set of actions that induce a null MEV state NS, we define the price of MEV as:
\begin{equation*}
    \text{PoMEV}(\mathbb G,n) = \frac{\max_{S\in SNE(\mathbb G)}C(S)}{\text{min}_{S\in \text{NS}}C(S)},
\end{equation*}
where $\text{min}_{a\in \text{NS}}C(a)$ is the minimum cost taken by extracting the MEV opportunity. 
\end{definition}

We argue that a more efficient mechanism to extract MEV opportunities (low Price of MEV) can have an important impact on blockchain stability, users utility, and consensus protocol\footnote{Note that we are assuming that the extractable value exists and will be extracted.}. In the case where the cost function is defined over pure strategies, we can extend the definition naturally over the mixed strategies, taking the expectancy of the outcomes. 

\section{Conclusions and Future Work}
In this work, we proposed measures to formally study the Nash equilibrium and negative externalities of different ordering mechanisms. We think that this is import due to the future changes on Ethereum mainnet about Proposer-Builder separation and MEV-boost \cite{mevboost2022}. We leave for future work to study and compute the price of MEV of popular proposed ordering mechanisms.
Also, in the future we will study the negative externalities of zero-sum MEV opportunities such as sandwich attacks, time bandit attacks, eclipse attacks, draining bot attacks\footnote{\hyperlink{https://github.com/Defi-Cartel/salmonella}{https://github.com/Defi-Cartel/salmonella}}, and censorship attacks. We conjecture that a domain with suboptimal block-space market design will lead to more block-space misuse (high price of MEV), raising the transactions fees of the underlying domain. The higher transactions fees will increase the direct payments of the miner (another form of MEV) and the market inefficiencies. This market inefficient will imply inefficient price discover, creating more internal and cross chain MEV opportunities. We call this effect the Circular forces of MEV \ref{fig:circular_mev}.
\begin{figure}[!h]
    \centering
    \includegraphics[scale=0.35]{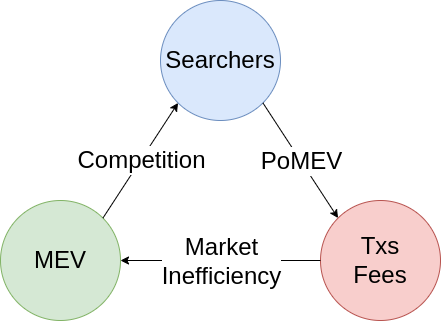}
    \caption{Circular forces of MEV}
    \label{fig:circular_mev}
\end{figure}

\printbibliography
\appendix
\section{Examples of ordering mechanisms}\label{appendix:order}
  For the sake of completeness, we will provide a list of examples of ordering mechanisms. They are induced in different domain and consensus protocol designs.
\begin{itemize}
    \item \textbf{Priority gas ordering mechanism }: Sequencers try to solve the KEV by using the greedy approximation algorithm that consists of ordering the transactions by gas price. In this case, if a player is trying to capture an MEV opportunity, it must monitor the mempool and choose an optimal gas price $m$. If a player is trying to front-run a transaction $\texttt{tx}$, with gas price $m$, it is enough to outbid it with $m+\eta$. In this setting, if the gas cost of exploiting this MEV opportunity is $g$, the block-space price of MEV of the uni-agent game is $1$.
    \item \textbf{Flashbots mechanism}: Searchers send bundles to the relayer through a private channel. The Flashbots relayer tries to build the block with the highest profits among all the blocks that can be constructed using the transactions in the public mempool and the Flashbots mempool of bundles. But the bundles have a number of allocation constraints that the Flashbots relayer must account for \cite{MaxEV}. In order to build the block with the highest profit, Flashbots (to our knowledge) uses a greedy approximation algorithm. As described in the Flashbots documentation, a bundle $B$ is ordered by effective gas price / bundle score, which is defined as,
    \begin{equation*}
        sc(B):= \frac{\Delta_{coinbase}+\sum_{tx\in B\setminus \mathcal T\mathcal X} g(tx)m(tx)}{\sum_{tx\in B} g(tx)}
    \end{equation*}
    where $\Delta_{coinbase}$ denotes the direct payment to the miner, $\mathcal T\mathcal X$ is the set of mempool transactions, $g(tx)$ is the gas used by $tx$ and $m(tx)$ is the gas price of $tx$.
    
    \item \textbf{Random ordering mechanism}:  The transactions included in the next block and the order of transaction execution are probabilistic with a uniform distribution. 
    \item \textbf{First input first output mechanism}: The transactions are ordered by the sequencer local's timestamps or by a pseudo-global timestamp such as the one mentioned in \cite{kelkar2020order}. In this sense, players with better geolocation and propagation algorithms will win the MEV game. However, in decentralized systems, this will depend on the leader geolocation that will change randomly per round. 
    \item \textbf{Dictatorship/Permissioned mechanism}: The sequencer has its own arbitrary ordering rule, prioritizing transactions of a fixed set of addresses. In this setting, players do not have a lot of freedom to interact or win the MEV opportunity. In other words, the sequencer will censor other players' transactions to prioritize its own extraction. Moreover, this potentially will induce inefficient market prices. However, the block-space price of anarchy is minimized since just one player is extracting it. This rule also models the situation where the miner captures the MEV opportunity, prioritizing its own profitable bundles.
    \item \textbf{Metadata mechanism}: Let $(\{0,1\}^n,\leq)$ be a total ordered set. Transactions and bundles can add a parameter \texttt{nonce}, giving them an associated hash identification. Then the bundles and transactions are ordered by hashes. For example, if a transaction $\texttt{tx}$ with $\texttt{nonce}$ tries to extract a back-running arbitrage opportunity, then a player will try to produce a transaction that extracts the opportunity nonce $\texttt{nonce'}<\texttt{nonce}$.\\
\end{itemize}
\begin{table}
\centering
\begin{tabular}{|l|c|c|c|c|c|}
\hline
&\multicolumn{1}{l|}{PGA} & \multicolumn{1}{l|}{FSS} & \multicolumn{1}{l|}{R.O.}                 & \multicolumn{1}{l|}{Perms.} & \multicolumn{1}{l|}{MP} \\ \hline
Ethereum (Geth)        & \cmark    & \xmark/$\sim$    & \xmark                     & \cmark               & \xmark                \\ \hline
Polygon         & \xmark    & \xmark    & \cmark                     & \cmark               & \xmark                \\ \hline
BSC             & \cmark    & \xmark    & \xmark                     & \cmark               & \xmark                \\ \hline
Avalanche       & \xmark    & \cmark    & \xmark                     & \xmark               & \cmark                \\ \hline
Arbitrum        & \xmark                        & \cmark    & \xmark                     & \cmark               & \cmark                \\ \hline
Shutter Network & \cmark                        & \xmark                      & \xmark                                        & \cmark                              & \cmark                \\ \hline
Solana          & \xmark    & \cmark    & \xmark                       & \cmark               & \xmark                \\ \hline
Flashbots (alpha-v0.6)      & \cmark    & \xmark    & \xmark/ & \cmark               & \cmark                \\ \hline
\end{tabular}
\caption{MEV games features different chains. Perms= Permissionless and MP = Mempool Privacy.}
\end{table}
\section{Bot example}
Before the pull request \cite{prfifo}, transactions with the same gas price were randomly ordered, creating incentives for searchers to spam transactions to capture an MEV opportunity. In the following, we will show how a particular MEV bot that captured back running opportunities, was responsible for consuming unnecessary  block space to increase its expected revenue. A lot of examples can be seen using tools such as \verb|mev-inspect-py|. Moreover, we could lower bound  the estimated block space price of MEV by $\approx 7$.
\begin{figure}[!h]
    \centering
    \includegraphics[scale=0.3]{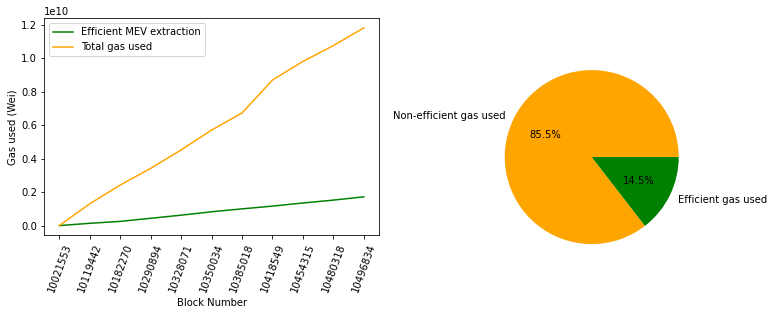}
    \caption{Total gas used and efficient gas used by bot MEV Bot 0x00...E9a}
\end{figure}

\section{Flashbots counter-example}\label{appendix:counter}

Let FBR($B_1,...,B_n$) be the revenue using the Flashbots greedy approximation algorithm. Let OTP($B_1,...,B_n)$ be the maximal revenue of the Flashbots combinatorial problem. 

\textbf{Claim}: The Flashbots combinatorial auction is not optimal. More specifically,
\begin{equation*}
    \text{inf}\{\frac{\text{FBR}(B_1,...,B_k)}{\text{OPT}(B_1,...,B_k)}:\text{ for }B_1,...,B_k \text{ bundles}\}\leq \frac{1}{\lfloor\frac{L}{g_{min}}\rfloor-1},
\end{equation*}
where $g_{min}$ is the minimal gas consumed by competing bundles and  $L$ is the gas limit of a block.

\textbf{Proof}: Let $B_1,...,B_k$ all the bundles such that, $B_1$ compete with $B_i$ for all $i\not=1$ and $B_i,B_j$ are pairwise non-competing bundles.  
Moreover, assume that $B_1$ has gas costs $L/k$ and effective gas bid $m+\varepsilon$ for $\varepsilon>0$ and all the other bundles have gas bid $m$. 
Then, the Flashbots algorithm outputs $B=\{B_1\}$, leaving to a sequencers' revenue of $m+\varepsilon$. On the other hand, the optimal valid block is $B=\{B_2,...,B_k\}$ with $m(k-1)$ revenue. The result follows using bundles with gas cost $g_{min}$.

\end{document}